\documentclass[draft,onecolumn]{IEEEtran}
\usepackage{lineno}

\usepackage{graphicx}
\usepackage{amssymb}
\usepackage{amsfonts}
\usepackage{graphicx,color}

\usepackage{arydshln}
\usepackage{multirow}
\usepackage{amsmath}

\usepackage{amsmath}
\usepackage{graphicx}
\usepackage{amsfonts}
\usepackage{amssymb}
\usepackage{mathrsfs}
\usepackage[mathscr]{euscript}
\usepackage[sans]{dsfont}
\usepackage{graphicx}
\usepackage{stmaryrd}
\usepackage{setspace}
\usepackage{tikz}

\usepackage{color}

\newenvironment{proof}[1][Proof]{\noindent\textbf{#1.} }{\ \rule{0.5em}{0.5em}}

\newcommand{\vv}{\pmb{v}}
\newcommand{\uu}{\pmb{u}}

\newcommand{\C}{\mathcal{C}}\usepackage{amsmath}
\usepackage{graphicx}
\usepackage{amsfonts}
\usepackage{amssymb}
\usepackage{mathrsfs}
\usepackage[mathscr]{euscript}
\usepackage[sans]{dsfont}
\usepackage{graphicx}
\usepackage{stmaryrd}
\usepackage{setspace}
\usepackage{tikz}

\newtheorem{theorem}{Theorem}

\newtheorem{corollary}[theorem]{Corollary}

\newtheorem{definition}[theorem]{Definition}
\newtheorem{example}[theorem]{Example}

\newtheorem{lemma}[theorem]{Lemma}

\newtheorem{proposition}[theorem]{Proposition}
\newtheorem{remark}[theorem]{Remark}

\begin{document}

\title{Robust low-delay Streaming PIR using convolutional codes 
}


\author{Julia Lieb,
 Diego Napp
 and~Raquel Pinto
\thanks{Julia Lieb is at the Department of Mathematics, University of Zurich, Switzerland e-mail: julia.lieb@math.uzh.ch.}
\thanks{Diego Napp is at the Department of Mathematics, University of Alicante, Spain
		 e-mail: diego.napp@ua.es.}
\thanks{Raquel Pinto is at the Department of Mathematics, University of Aveiro, Portugal e-mail: raquel@ua.pt.}

}


%



\markboth{Journal of \LaTeX\ Class Files,~Vol.~-, No.~-, -~201-}%
{Shell \MakeLowercase{\textit{et al.}}: Bare Demo of IEEEtran.cls for IEEE Journals}
%

\maketitle

\begin{abstract}

In this paper we investigate the design of a low-delay robust streaming PIR scheme on coded data that is resilient to unresponsive or slow servers and can privately retrieve streaming data in a sequential fashion subject to a fixed decoding delay. We present a scheme based on convolutional codes and the star product and assume no collusion between servers. In particular we propose the use of convolutional codes that have the maximum distance increase, called Maximum Distance Profile (MDP). We show that the proposed scheme
can deal with many different erasure patterns.
\end{abstract}

\section{Introduction}

Video traffic has had an explosive growth and it is expected to keep its exponential growth in the coming years \cite{Cisco2017}. Service providers for real-time video streaming are typically hosted in a public cloud, with multiple servers in different data centers, e.g. Google Cloud, Amazon CloudFront and Microsoft Azure. These cloud services aim for private and low-latency communications.

The problem of Private Information Retrieval (PIR) has attracted a lot of attention in the recent years and studies how to retrieve a file from a storage system without revealing the desired file to the servers. It was initially addressed for replicated files \cite{Chor1998} and recently for coded files \cite{BaUl2018,HollandiPIR2017,UmbertoISIT2019,Shah2014,TaRo2016}. In this last setting, the general model of the information theoretic PIR problem is as follows. Let a coded database be contained in $n$ servers storing $m$ files and assume that the user knows the content of the servers. Each file is coded and stored independently using the same code and the user wants to retrieve a particular file from the database with zero information gain from the servers, i.e., the user wants information theoretic privacy \cite{TaRo2017}. Recently, the literature on PIR models has grown considerably with extensions for more general PIR models with several additional constraints. Most of the efforts in private retrieval have focused on efficient schemes that optimize different metrics, such as communication cost or rate. However, in many cases some of the servers may be busy and do not respond within a desirable time frame or network failures may occur. For this reason, new robust schemes were proposed in order to deal with such scenarios \cite{TaRo2017,Hollanti2019} adding redundancy to tolerate certain missing files from the servers. PIR schemes on coded data for Byzantine or unresponsive servers were presented in \cite{Zhang2017PrivateIR,Hollanti2019}. These schemes are suited for retrieving one single file and therefore use block codes.  In \cite{Hollandi_Convolutional_2018} a scheme for sequential retrieving was proposed but again for a given set of files of fixed size and assuming that \textit{all} the responses of the servers are lost at the same time instant. The case of a non-bursty channel is also considered in this paper but only using unit memory convolutional codes. However, to the best of the author's knowledge, the problem of low-delay private retrieval of a stream of files (of undetermined length) with \textit{some} low or unresponsive servers remains unexplored.

In this work we investigate this more general problem and propose a novel robust scheme to deal with low-delay streaming retrieval of files from $n$ servers in the presence of possible unresponsive servers
by using Maximum Distance Profile (MDP) convolutional codes. 
This class of codes is suitable for low-delay streaming applications as they possess optimal error-correcting capabilities within a decoding window, see \cite{ba15b,to12}. One of the advantages of using convolutional codes over block codes is the sliding window flexibility that allows to select different decoding windows according to the erasure pattern. We show how to take advantage of this property to provide robust PIR in this context. We present a scheme that is able to stream files consisting of many stripes in the presence of erasures without assuming any particular structure in the sequence of erasures. The model in \cite{Hollandi_Convolutional_2018} treated burst erasure channels using general convolutional codes and the non-bursty channel case was treated using unit memory convolutional codes. Unit memory codes are restricted to store only what occurred in the previous instant and therefore are far from optimal for low-delay applications when the given delay constraint is larger than one. Note also that when only burst erasure channels are assumed, there exist concrete constructions of convolutional codes that are optimal in such a context \cite{ba15b,ba13,Martinian2007}. In this work we extend this thread of research and consider a non necessarily bursty channel using convolutional codes with no restriction in the memory, namely, MDP convolutional codes. In contrast to \cite{Hollandi_Convolutional_2018}, where the response of the servers is built in a convolutional fashion but the storage code is still a block code, we also use a convolutional code to store the files on the servers.

\section{Preliminaries}

In this section we recall basic material and introduce the definitions needed for this work, including the notion of convolutional code and superregular matrix.  Let
$\mathbb F =\mathbb F_q $ be a finite field of size $q$ and $\mathbb F[z]$ be the ring of polynomials with coefficients in $\mathbb F$.

\begin{definition}
	An $\mathbf{[n,k]}$-\textbf{block code} $\C$ is a $k$-dimensional subspace of $\mathbb F^n$, i.e., there exists $G\in\mathbb F^{k\times n}$ of full row rank matrix such that
	$$\C=\{\vv\in\mathbb F^n\ |\ \vv=\uu G\ \text{for some}\ \uu\in\mathbb F^k\}.$$
	$G$ is called \textbf{generator matrix} of the code and is unique up to left multiplication with an invertible matrix $U\in Gl_k(\mathbb F)$. Furthermore, $u\in\mathbb F^k$ is called \textbf{message vector} and the elements $\vv\in\C$ are called \textbf{codewords}.
\end{definition}


Convolutional codes process a continuous sequence of data instead of blocks of fixed vectors as done by block codes.
%
If we introduce a variable $z$, called the \emph{delay operator}, to indicate the time instant in which each information arrived or each codeword was transmitted, then we can represent the sequence message $(\vv_{0},\vv_{1} , \cdots , \vv_l )$ as a polynomial vector $\vv(z)= \vv_{0}+\vv_{1} z^{} + \cdots + \vv_l z^l $.
Formally, we can define convolutional codes as follows.

A rate $k/n$ \textbf{convolutional code} $\mathcal{C}$ \cite{to12} is an $\mathbb{F}[z]$-submodule of $\mathbb{F}[z]^n$ with rank $k$ given by
$$
\mathcal{C}=\text{im}_{\mathbb{F}[z]}G(z)=\{\vv(z) \in \mathbb{F}[z]^n \  | \ \vv(z)= \uu(z)G(z), \text{ with }  \uu(z)\in \mathbb{F}[z]^k\}
$$
where $G(z)\in \mathbb{F}[z]^{k\times n}$ is a matrix, called generator matrix, that is \emph{basic}, i.e., $G(z)$ has a polynomial right inverse.

Note that if 
$\vv(z)=\uu(z)G(z)$, with
\begin{align*}
\uu(z) &= \uu_0 + \uu_{1}z + \uu_{2} z^{2} + \cdots \ \text{ and } \
G(z) = \sum_{j =0}^\mu G_{j} z^{j}
\end{align*}
then, 
\begin{align*}
\lefteqn{\vv_0 + \vv_{1}z+\vv_{2} z^{2} + \cdots}\\
& = \uu_{0} G_{0}+ \left(\uu_{1}G_{0}+ \uu_{0} G_{1}\right)z+ \left(\uu_{2} G_{0}+\uu_{1} G_{1}+\uu_{0} G_{2}\right)z^2+\cdots
\end{align*}

The maximum degree of all polynomials in the $j$-th row of $G(z)$ is denoted by $\delta_{j}$. The degree $\delta $ of $\mathcal{C}$ is defined as the maximum degree of the full size minors of $G(z)$. We say that $\mathcal{C}$ is an $(n,k,\delta)$ convolutional code \cite{mc98}.
Important for the performance of a code in terms of error-free decoding is the (Hamming) distance between two codewords. In the case of convolutional codes, the most relevant notion of distance for low-delay decoding is the column distance that can be defined as follows.

The $\mathbf{j}$\textbf{-th column distance} \cite{jo99} is defined as
\[
d_{j}^{c}\left(\mathcal{C}\right) =
\min\left\{  \text{wt} \left(  \vv_{[0,j]}(z)\right) |\ \vv(z) \in\mathcal{C} \text{ and } \vv_{0} \neq \mathbf{0} \right\}, %
\]
where $ \vv_{[0,j]}(z)=\vv_{0}+\vv_{1} z + \cdots + \vv_j z^j $ represents the $j$-th truncation of the codeword $ \vv(z) \in\mathcal{C}$ and
\[
\text{wt}(\vv_{[0,j]}(z))=\text{wt}(\vv_{0})+\text{wt}(\vv_{1})+ \cdots + \text{wt}(\vv_j)
\]
where $\text{wt}(\vv_{i})$ is the Hamming weight of $\vv_{i}$, which determines the number of nonzero components of $\vv_{i}$, for $i=1, \ldots, j$.
For simplicity, we use $d_j^c$ instead of $d_j^c(\mathcal{C})$.

The $j$-th column distance is upper bounded \cite{gl03} by
\[
d_j^c \leq (n-k)(j+1)+1,
\]
and the maximality of any of the column distances implies the maximality of all the previous ones, that is, if $d_j^c = (n-k)(j+1)+1$ for some $j$, then $d_i^c = (n-k)(i+1)+1$ for all $i \leq j$.
The value
\begin{equation}\label{MDPvalue}
L=\left \lfloor \frac{\delta}{k} \right \rfloor + \left \lfloor \frac{\delta}{n-k} \right \rfloor
\end{equation}
is the largest value for which the bound can be achieved and an $(n,k,\delta)$ convolutional code $\mathcal{C}$ with $ d_L^c= (n-k)(L+1)+1$ is called a \emph{maximum distance profile} (MDP) code  \cite{gl03}. Hence, MDP codes have optimal error correcting capabilities within time intervals and therefore are ideal for low delay correction. In this work we shall assume that the retrieval must be performed within a given delay constraint $\Delta \leq L$, see \cite{cassuto17,ba15b}. 

%


Assume that
\[
G(z)=\sum_{j=0}^{\mu}G_{j}z^{j},G_{j}\in\mathbb{F}^{k\times n},G_{\mu}\neq 0,
\]
and consider the associated \textbf{sliding  matrix}
\begin{equation} \label{SliGen}
G_{j}^c=\left(
\begin{array}
[c]{cccc}%
G_{0} & G_{1} & \cdots & G_{j}\\
& G_{0} & \cdots & G_{j-1}\\
&       & \ddots & \vdots\\
&       &         & G_{0}%
\end{array}
\right)
\end{equation}
with $G_j=0$ when $j>\mu$, for $j\in\mathbb{N}$.

%
%
%
%
%

\begin{theorem}[Theorem~$2.4$ in \cite{gl03}]\label{slid}
	Let $G_{j}^c$ be the matrices defined in (\ref{SliGen}). Then the following statements are equivalent:
	\begin{enumerate}
		\item $d_{j}^c=(n-k)(j+1)+1$;
		\item  every $(j+1)k\times (j+1)k$ full size minor of $G_{j}^c$ formed from the columns with the indices $1 \leq t_1 \leq \cdots \leq t_{(j+1)k}$, where $t_{ik+1} \leq in$, for $i=1,2,\ldots,j$ is nonzero;
	\end{enumerate}
	In particular, when $j=L$, $\mathcal{C}$ is an MDP code.
\end{theorem}

\begin{theorem}\label{window}\cite[Theorem 3.1]{to12}
Let $\C$ be an $(n,k,\delta)$ MDP convolutional code.
If $d_j^c =(n-k)(j + 1) + 1$  and in any sliding window of length $(j + 1)n$ at most $(j+1)(n-k)$ erasures occur in a transmitted sequence, then complete recovery is possible.
\end{theorem}

Considering the proof of this theorem, one sees that the recovery is even possible within a delay of $j$ windows of size $n$ and that the given condition for complete recovery is only sufficient but not necessary.\\

We will develop a PIR scheme in which the star product of certain block codes plays an important role.

\begin{definition}\label{star}
	The \textbf{star product} of two vectors $v,w\in\mathbb F^n$ is defined as $v\ast w=(v_1w_1,\hdots,v_nw_n)$. The \textbf{star product} of two block codes $\mathcal{C}, \mathcal{D}\subset\mathbb F^n$ is defined as $\mathcal{C}\ast\mathcal{D}=\langle c\ast d\ |\ c\in\mathcal{C},\ d\in\mathcal{D}\rangle$.
\end{definition}

Star product PIR was first introduced in \cite{HollandiPIR2017}. The main idea of this scheme is to design  the queries to the different servers in such a way that if the responses are formed as inner products of the query and the stored information, then the total response is a codeword of a certain star product code with some error, where the error contains the information one is interested in.
In \cite{Hollandi_Convolutional_2018} this scheme was adopted forming the responses in a convolutional way. In the following section, we present a star product scheme where the responses as well as the storage code are convolutional.

\section{Streaming PIR scheme}

We have $m$ sequences of files $(X_s^{i})_{s\in\mathbb N}$ with $X_s^{i}\in\mathbb F^k$ for $i=1,\hdots, m$ and $s\in\mathbb N$.  These are encoded with an $(n,k,\delta)$ MDP convolutional code $\mathcal{C}$ with generator matrix $G(z)=\sum_{r=0}^{\mu}G_rz^{r}$ to obtain the sequences of files $(Y_t^{i})_{t\in\mathbb N}$ with $Y_t^{i}=\sum_{r+s=t}X_s^{i}G_r\in\mathbb F^n$ for $i=1,\hdots, m$ and $t\in\mathbb N$ where we set $G_r=0$ for $r>\mu$. Moreover, we have $n$ servers and for $j=1,\hdots, n$, we store the $j$-th component $Y_{t,j}^{i}$ of each vector $Y_t^{i}$ (for $i=1,\hdots,m$, $t\in\mathbb N$) on server number $j$. Furthermore, we assume that $(\mu+2)k\leq n$ and that for $f=1,\hdots,\mu$, $\begin{pmatrix}G_0\\ \vdots\\ G_f\end{pmatrix}$ is the generator matrix of an $[n,(f+1)k]$ MDS block code denoted by $\mathcal{C}_f$ . We will present a construction for an $(n,k,\delta)$ MDP convolutional code with these properties later in this paper. It holds $Y_t^{i}\in\mathcal{C}_f$ for all $f\in\{t,\hdots,\mu\}$ and $Y_t^{i}\in\mathcal{C}_{\mu}$ for $t\geq\mu$. Thus, we set $\mathcal{C}_f=\mathcal{C}_{\mu}$ for $f\geq\mu$.

The user wants to stream the sequence $(X_s^{i})_{s\in\mathbb N}$ for some $i$ without the servers knowing $i$, i.e. without that the servers know which sequence he or she is streaming. For our PIR scheme we assume that there is no collusion between the servers (i.e. the number of colluding servers, usually denoted by $t$ in the literature is equal to $1$).\\
Set $d=[1\ \hdots\ 1]\in\mathbb F^n$, let $\mathcal{D}$ be the $[n,1]$ block code generated by $d$ and $D\in\mathbb F^{(\mu+1)m\times n}$ be a matrix whose rows are constituted by $(\mu+1)m$ random codewords of  $\mathcal{D}$ (i.e. multiples of $d$). For a subset $J\subset\{1,\hdots,n\}$, we denote by $E\in\mathbb F^n$ the vector with entries $E_j:=\begin{cases}1 & \text{if}\  j\in J\\  0 & \text{otherwise}\end{cases}$ and we denote by $e_j$ the $j$-th standard basis vector of $\mathbb F^{(\mu+1)m}$.\\
For $j=1,\hdots,n$, we send the following query $q_j^{i}$ to server $j$:
\begin{align}
q_j^{i}=D_{\cdot,j}+E_j\sum_{l=0}^{\mu}e_{lm+i}\in\mathbb F^{(\mu+1)m\times 1}
\end{align}
where $D_{\cdot,j}$ denotes the $j$-th column of $D$. We write $q_j^i=\begin{pmatrix}q_{j,1}^i\\ \vdots\\ q_{j,\mu+1}^i\end{pmatrix}$ with $q_{j,k}^i\in\mathbb F^{m\times 1}$ for $k=1,\hdots,\mu+1$ and $Y_{t,j}:=(Y^1_{t,j}, Y^2_{t,j},\hdots, Y^m_{t,j})\in\mathbb F^{m}$.\\
The response of server $j$ at time $t\in\mathbb N$ is
\begin{align}
r_{t,j}^{i}=\sum_{k+r-1=t}\langle q_{j,k}^{i},Y^{\top}_{r,j}\rangle\in\mathbb F
\end{align}
where $Y_{r,j}=0$ for $r\not\in\mathbb N$.\\
Hence the total response at time $t\in\mathbb N$ is given by
\begin{align}
r_t^{i}&=[r_{t,1}^{i},\hdots,r_{t,n}^{i}]=\nonumber\\
&=\underbrace{[D_{1,1}Y_{t,1}^{1},\hdots, D_{1,n}Y_{t,n}^{1}]}_{\in\mathcal{D}\ast\mathcal{C}_t}+\underbrace{[D_{2,1}Y_{t,1}^{2},\hdots, D_{2,n}Y_{t,n}^{2}]}_{\in\mathcal{D}\ast\mathcal{C}_t}+\cdots+\nonumber\\
&+\underbrace{[D_{(\mu+1)m,1}Y_{t-\mu,1}^{m},\hdots, D_{(\mu+1)m,n}Y_{t-\mu,n}^{m}]}_{\in\mathcal{D}\ast\mathcal{C}_t}+diag(E)\sum_{l=0}^{\mu}Y^{i}_{t-l}
\end{align}
where diag(E) denotes the diagonal matrix with diagonal entries equal to the entries of the vector $E$.

By Definition \ref{star} and the definition of the code $\mathcal{D}$ the star product code $\mathcal{D}\ast\mathcal{C}_t$ is equal to the MDS code $\mathcal{C}_t$. As $\mathcal{D}\ast\mathcal{C}_t$ is a linear code, any sum of codewords is again a codeword. Hence, the response has the form
\begin{align}\label{2parts}
r_t^{i}=c_t+diag(E)\sum_{l=0}^{\mu}Y^{i}_{t-l}
\end{align}
for some $c_t\in\mathcal{C}_t$.

We assume that it is possible that some parts of the response at time $t$ get lost during transmission and could not be received. Hence the vector $r_t^{i}$ could have some erased components. We denote by $T_t\subset\{1,\hdots,n\}$ the set that consists of the positions of the erased components of the vector $r_t^{i}$.

\begin{lemma}
	If $|T_t\cup J|<n-k(\min\{t,\mu\}+1)+1$, the user is able to obtain the vector $diag(E)\sum_{l=0}^{\mu}Y^{i}_{t-l}$. In particular, this is true if $|J|+n_t<n-k(\min\{t,\mu\}+1)+1$, where $n_t$ is the number of erased components of the vector $r_t^{i}$.
\end{lemma}

\begin{proof}
Using equation \eqref{2parts} and the definition of the vector $E$, we apply erasure decoding in the $[n,(\min\{t,\mu\}+1)k]$ MDS code $\mathcal{C}_t$ to the vector $r_t^{i}$ where the set of erasures is the union of  $T$ and $J$. The lemma follows from the fact that an $[n,(\min\{t,\mu\}+1)k]$ MDS code could correct  any set of erasures whose cardinality is smaller than the minimum distance $n-k(\min\{t,\mu\}+1)+1$ of the code.
\end{proof}

For each $t\in\mathbb N$ for which the condition of the preceding lemma is not fulfilled we are not able to obtain  $diag(E)\sum_{l=0}^{\mu}Y^{i}_{t-l}$. Therefore, we define
\begin{align}
diag_t(\hat{E})=\begin{cases} diag(E) & \text{if}\ |T_t\cup J|<n-k(\min\{t,\mu\}+1)+1\\
0_n & \text{otherwise}\end{cases}
\end{align}
where $0_n$ denotes the $n\times n$ zero matrix.

It remains to show how to obtain the desired sequence of files $(X_s^{i})_{s\in\mathbb N}$ from the sequence $(diag_t(\hat{E})\sum_{l=0}^{\mu}Y^{i}_{t-l})_{t\in\mathbb N}$. With the definitions $\hat{r}^{i}_t:=\sum_{l=0}^{\mu}Y^{i}_{t-l}$ and
\begin{align}\label{sliding}
\tilde{\mathcal{G}}:=\left[\begin{array}{ccccccccc}G_0 & G_0+G_1 & \cdots & \sum_{r=0}^{\mu}G_r & \sum_{r=1}^{\mu}G_r  & \cdots & G_{\mu} & & \\
& G_0 & G_0+G_1 & \cdots & \sum_{r=0}^{\mu}G_r & \sum_{r=1}^{\mu}G_r  & \cdots & G_{\mu} & \\
&\qquad \  \ddots & \ddots & & \ddots & \ddots & &\qquad \ \ddots\end{array}\right]
\end{align}
one obtains
\begin{align}\label{dec}
[\hat{r}^{i}_1, \hat{r}^{i}_2, \hdots ]
=[X_1^{i},X_2^{i},\hdots ]\cdot\tilde{\mathcal{G}}
\end{align}
Denote by $I_k\in\mathbb F^{k\times k}$ the identity matrix and set $U\!:=\!\!\left[\begin{array}{cccc}I_k & \cdots & I_k & \\
&I_k & \cdots & I_k  \\
&\qquad \  \ddots &  & \qquad \ \ddots\end{array}\right]$ where each block of $k$ rows of $U$ contains $\mu+1$ identity matrices. Then, one has

\begin{align}\label{u}
\tilde{\mathcal{G}}=U\cdot\underbrace{\left[\begin{array}{cccccc}G_0 & G_1 & \cdots & G_{\mu} & & \\
	& G_0 & G_1 & \cdots  & G_{\mu} & \\
	&\qquad \  \ddots & \qquad\ddots & &\qquad \ \ddots\end{array}\right]}_{:=\mathcal{G}}.
\end{align}
Therefore, one obtains the following lemma.

\begin{lemma}
	The column distances of the convolutional code $\tilde{\mathcal{C}}$ with generator matrix $\tilde{G}(z)=\sum_{b=0}\tilde{G}_bz^{b}$ where $\tilde{G}_b=\sum_{r=0}^{\mu}G_{b-r}$ are equal to the column distances of $\mathcal{C}$.
\end{lemma}

\begin{proof}
First note that the matrix $\tilde{\mathcal{G}}$ defined in \eqref{sliding} is the sliding generator matrix of $\tilde{\mathcal{C}}$. Denote by $\tilde{d}_j^{c}$ the $j$-th column distance of the code $\tilde{\mathcal{C}}$ and by $U_j$ the matrix that consists of the first $k(j+1)$ rows and the first $k(j+1)$ columns of the matrix $U$. Then, it holds
\begin{align}
\tilde{d}_j^{c}&=\min_{X^{i}_1\neq 0}\left(wt\left([X^{i}_1, \hdots, X^{i}_{j+1}]\cdot\tilde{G}_j^c\right)\right)\stackrel{\eqref{u}}{=}\min_{X^{i}_1\neq 0}\left(wt\left([X^{i}_1, \hdots, X^{i}_{j+1}]U_j\cdot G_j^c\right)\right)\nonumber\\
&=\min_{\hat{X}^{i}_1\neq 0}\left(wt\left([\hat{X}^{i}_1, \hdots, \hat{X}^{i}_{j+1}]\cdot G_j^c\right)\right)=d_j^c
\end{align}
\end{proof}

Hence, we could use equation \eqref{dec} to obtain $[diag_1(\hat{E})\hat{r}^{i}_1, diag_2(\hat{E})\hat{r}^{i}_2, \hdots ]$ via erasure decoding with an MDP convolutional code where the set of positions of the total erasures denoted by $T$ has the form $T=\bigcup_{t\in\mathbb N}S_t$ with
\begin{align}\label{T}
S_t=\begin{cases}\underbrace{\{T_t+(t-1)n\}}_{\text{transmission erasures}}\cup\underbrace{(\{1+(t-1)n,\hdots,n+(t-1)n\}\setminus\{J+(t-1)n\})}_{\text{erasures caused by the multiplication with}\ diag(E)} & \text{if}\ diag_t(\hat{E})\neq 0\\
\vspace{0.1mm}\\
\{1+(t-1)n,\hdots,n+(t-1)n\} & \text{otherwise}\end{cases}
\end{align}
\normalsize
where for $J=\{j_1,\hdots,j_{|J|}\}$, the set $\{J+(t-1)n\}$ is defined as\\ $\{j_1+(t-1)n,\hdots,j_{|J|}+(t-1)n\}$ and $\{T_t+(t-1)n\}$ should be defined analogous.
Hence, using also Theorem \ref{window}, we get the following theorem.

\begin{theorem}
Assume that $\Delta\leq L=\lfloor\frac{\delta}{n-k}\rfloor+\lfloor\frac{\delta}{k}\rfloor$. If the set of erasures $T$ given in \eqref{T} is such that in every sliding window of the sequence $(r_t^{i})_{t\in\mathbb N}$ of size $(\Delta+1)n$ there are not more than $(\Delta+1)(n-k)$ erasures, then one could obtain the desired sequence of files $(X_s^{i})_{s\in\mathbb N}$ from the sequence $(diag_t(\hat{E})\sum_{l=0}^{\mu}Y^{i}_{t-l})_{t\in\mathbb N}$ within time delay $\Delta$, i.e. one could privatly obtain the sequence of files $(X_s^{i})_{s\in\mathbb N}$ within time delay $\Delta$.
\end{theorem}

From this theorem we could deduce which erasure patterns we can correct for sure with our proposed scheme.

\begin{corollary}\label{er}
	With the proposed scheme private reception within time delay $\Delta\leq L$ is possible if for $t\in\mathbb N$, there are not more than $n-k(\min\{\mu,t\}+1)-|J|$ transmission erasures in positions $\{1+(t-1)n,\hdots,n+(t-1)n\}\setminus\{J+(t-1)n\}$  of the sequence of responses $(r^{i}_t)_{t\in\mathbb N}$ and in every sliding window of this sequence of length $(\Delta+1)n$ there are not more than $(\Delta+1)(n-k)$ transmission erasures in positions $\{1,\hdots,(\Delta+1)n\}\cap\bigcup_{t\in\mathbb N}\{J+(t-1)n\}$.
\end{corollary}

Finally, we have to choose the cardinality of the set $J\subset\{1,\hdots,n\}$. Then, the set $J$ is chosen randomly with this fixed cardinality. If the cardinality of $J$ is larger, this leads to more erasures for $\mathcal{C}_t$ to correct. But in turn if the cardinality of $J$ is smaller, this leads to more erasures for $\mathcal{C}$ to correct.
To balance this somehow, we want to determine $|J|$ such that the number of erasures one could correct in positions $\{1+(t-1)n,\hdots,n+(t-1)n\}\setminus\{J+(t-1)n\}$ is approximately the same as the number of erasures one could correct in positions $\{1+(t-1)n,\hdots,n+(t-1)n\}\cap\{J+(t-1)n\}$. We denote this number of erasures by $n_t$. This approach leads to the following equations:
\begin{align}
n_t&\leq n-k(\min\{\mu,t\}+1)-|J|\quad \text{and}\\
n_t&\leq n-k-(n-|J|) =|J|-k. \label{ermdp}
\end{align}
This implies
\begin{align}
n_t-k\leq|J|\leq  n-k(\min\{\mu,t\}+1)-n_t
\end{align}
and consequently,
\begin{align}
n_t\leq\frac{1}{2}(n-k(\min\{\mu,t\}+2)).
\end{align}
Having equality in this last equation, implies $|J|=\frac{1}{2}(n-k\min\{\mu,t\})$. However, we need $|J|$ to be an integer and independent of $t$. As depending on the erasure pattern, the MDP convolutional code $\mathcal{C}$ might be able to correct more erasures than \eqref{ermdp} indicates, we propose to rather choose $|J|$ smaller,
which finally leads to
\begin{equation}\label{J}
|J|=\lfloor \frac{1}{2}(n-k\mu)\rfloor.
\end{equation}

Of course, depending on the erasures that occur during transmission, other choices for $|J|$ could lead to a better performance. However, as we do not know the erasure pattern before transmission and we have to choose $J$ before, we cannot adapt $J$ corresponding to the erasure pattern but have to choose it in a way that the numbers of channel erasures our codes $\mathcal{C}_t$ and $\mathcal{C}$  are able to tolerate are balanced.




Note that we can correct more erasures in $r_t^{i}$ if $t$ is small (as the code $\mathcal{C}_t$ has a larger minimum distance if $t$ is small). This means that we could tolerate slightly more erasures at the beginning of the stream than at the end.

In the following, we illustrate the erasure correcting capability of our scheme with the help of two examples.

\begin{example}
	Let $n=6$, $k=1$ and $\mu=2$. This implies $\delta=2$ and $L=2$, i.e. $\mathcal{C}$ is an $(6,1,2)$ MDP convolutional code that could recover all erasures patterns for which in each sliding window of size $18$ there are not more than $15$ erasures. We assume $\Delta=L$. Moreover, according to equation \eqref{J}, we have $|J|=2$. We  illustrate one window of the response sequence $(r^{i}_t)_{t\in\mathbb N}$ in the following figure, where the squares with content $j$ denote the positions of the set $J$:\\
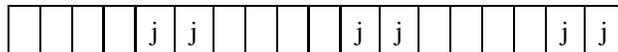
\begin{center}
	\begin{tabular}{|  c| c | c  | c  |c |c | c |c |c|c|c|c|c|c|c|c|c|c|}
		\hline
		&  &  &  &j & j &  &  & &  & j & j & &  &  &  &j & j\\
		\hline
	\end{tabular}\\
\end{center}	
	
\vspace{.5cm}
	
	According to Theorem \ref{er} we are able to recover $2$ erasures in the first $4$ positions with erasure decoding in $\mathcal{C}_1$. Moreover, $\mathcal{C}_2$ and  $\mathcal{C}_3$ are both able to correct $1$ additional erasure. Finally, the convolutional code $\mathcal{C}$ is able to correct $3$ erasures in the positions in which we have a $j$. To count the total number of erasures as well as the number of erasure patterns that can be corrected (assuming that erasures occur independently of each other), we have to distinguish two cases.
	
	For the first case, we assume that the erasure pattern allows decoding with $\mathcal{C}_1$, $\mathcal{C}_2$ and $\mathcal{C}_3$. Hence, we are able to correct up to $7$ erasures in $18$ positions. Moreover, if we assume that the erasures occur independently of each other, we could correct $\left(\sum_{i=0}^{2}\binom{4}{i}\right)\cdot 5\cdot 5\cdot\left(\sum_{i=0}^{3}\binom{6}{i}\right)=10175$ different erasure patterns.
	
	For the second case, we assume that the erasure pattern is such that there exists $t\in\{1,2,3\}$ such that decoding with $\mathcal{C}_t$ is not possible, i.e. the $t$-th window of size $n=6$ has to be considered as completely lost for $\mathcal{C}$. In order that recovery is still possible, decoding with $\mathcal{C}_s$ for $s\neq t$ has to be possible and only one additional erasure in the positions in $J$ outside the completely erased window can be tolerated. Thus, for $t=1$ the maximal number of erasures that can be corrected is $9$ and the number of correctable erasure patterns equals $625$. For $t\neq 1$, the maximum number of erasures that can be corrected is $10$ and the number of correctable erasure patterns equals $2750$.
	
	Summimg up, considering all cases, one gets that there are $13550$ erasure patterns that we can correct.

	If one would choose $|J|=1$, correction is not possible anymore if one complete window of size $n$ is lost. We would still be able to correct $7$ erasures but all these erasures have to be in positions

	$$
\bigcup_{t=1}^{L+1}\left(\{1+(t-1)n,\hdots,n+(t-1)n\}\setminus\{J+(t-1)n\}\right)
$$
whereas no erasures in positions
 $$
 \bigcup_{t=1}^{L+1}\left(\{1+(t-1)n,\hdots,n+(t-1)n\}\cap\{J+(t-1)n\}\right),
 $$
could be corrected. Counting the number of erasure patterns that we are able to correct under the assumption of independent erasures, we get $6656$.

	If one would choose $|J|=3$, there are three cases to distinguish. For the first case, assume that no window of size $n$ is completely lost for recovery with $\mathcal{C}$. Then, we could again correct $7$ erasures but only $1$ of these erasures can have a position in
	$$
\bigcup_{t=1}^{L+1}\left(\{1+(t-1)n,\hdots,n+(t-1)n\}\setminus\{J+(t-1)n\}\right).
$$
The number of erasure patterns that could be corrected is $1864$.

 For the second case, assume that correction with $\mathcal{C}_t$ is not possible for exactly one $t\in\{1,2,3\}$. For $t=1$, one could correct up to $9$ erasures and $168$ erasure patterns, for $t\neq 1$, up to $10$ erasures and $2352$ erasure patterns.

 For the third case, assume that correction with $\mathcal{C}_{t}$ is not possible for exactly two values $t\in\{1,2,3\}$, denoted by $t_1$ and $t_2$. If $1\in\{t_1,t_2\}$, one could correct up to $12$ erasures and $56$ erasure patterns, for $1\notin\{t_1,t_2\}$, up to $13$ erasures and $196$ erasure patterns.

 Hence the total number of erasure patterns that could be corrected is $4636$.
This illustrates that our choice of $J$ is optimal if we assume the erasures to occur independently of each other.

	Finally, we want to consider, how many erasures we can correct in a larger window and choose a window of size 24, which is illustrated as follows:\\
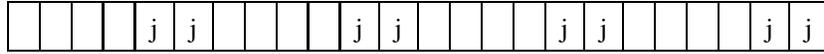
	\begin{center}
	\begin{tabular}{|  c| c | c  | c  |c |c | c |c |c|c|c|c|c|c|c|c|c|c|c|c|c|c|c|c|}
		\hline
		&  &  &  &j & j &  &  & &  & j & j & &  &  &  &j & j& &  &  &  &j & j\\
		\hline
	\end{tabular}
	\end{center}	
	
\vspace{.5cm}

	According to Theorem \ref{er} we are able to recover $2$ erasures in the first $4$ positions with erasure decoding in $\mathcal{C}_1$ and $3$ additional erasures with $\mathcal{C}_t$ for $t\geq 2$. The convolutional code $\mathcal{C}$ is able to correct up to $5$ erasures in the positions with $j$. Under the assumption that decoding with $\mathcal{C}_t$ is possible for $t=1,\hdots,4$, we are able to correct up to $10$ erasures. If decoding is not possible for exactly one $t$, one can correct up to $12$ erasures if $t=1$ and up to $13$ erasures if $t\neq 1$. If decoding is not possible for exactly two of the star product codes, recovery is only possible if this happens for $t=1$ and $t=4$, in which case up to $15$ erasures can be corrected.
	
	If one would choose $|J|=1$, we would only be able to correct $7$ erasures and all these erasures have to be in positions in
$$
\bigcup_{t=1}^{L+1}\left(\{1+(t-1)n,\hdots,n+(t-1)n\}\setminus\{J+(t-1)n\}\right).
$$
	If one would choose $|J|=3$, one has to distinguish four cases. Under the assumption that decoding with $\mathcal{C}_t$ is possible for $t=1,\hdots,4$, we are able to correct up to $10$ erasures in total but only $1$ of these erasures could be in
$$
\bigcup_{t=1}^{L+1}\left(\{1+(t-1)n,\hdots,n+(t-1)n\}\setminus\{J+(t-1)n\}\right).
$$
If decoding is not possible for exactly one $t$, one can correct up to $12$ erasures if $t=1$ and up to $13$ erasures if $t\neq 1$. If decoding is not possible for exactly two of the star product codes and $\mathcal{C}_1$ is among them, one could correct up to $15$ erasures and if $\mathcal{C}_1$ is not among them, one could correct up to $16$ erasures. If decoding is not possible for exactly three of the star product codes, one could correct up to $18$ erasures (but there are only two erasure patterns for this scenario).
\end{example}

\begin{example}
	Let $n=10$, $k=2$ and $\mu=2$. This implies (if we use for $\mathcal{C}$ the construction presented in the next section where $G_{\mu}$ is full rank) $\delta=4$ and $L=2$, i.e. $\mathcal{C}$ is an $(10,2,4)$ MDP convolutional  code that can recover all erasures patterns for which in each sliding window of size $30$ there are not more than $24$ erasures. We assume $\Delta=L$. Moreover, according to equation \eqref{J}, we have $|J|=3$.
	
	
	According to Theorem \ref{er} we are able to recover $3$ erasures in the first $7$ positions of the response sequence $(r^{i}_t)_{t\in\mathbb N}$ with erasure decoding in $\mathcal{C}_1$. Moreover, $\mathcal{C}_2$ and  $\mathcal{C}_3$ are both able to correct $1$ additional erasure. Finally, the convolutional code $\mathcal{C}$ is able to correct $3$ erasures in the positions covered by one of the sets $\{J+(t-1)n\}$. In total, we are able to correct $8$ erasures in $30$ positions in the case that correction with all $\mathcal{C}_t$ is possible, up to $12$ erasures in the case that (only) the first window of size $n=10$ is lost completely and up to $14$ erasures in the case that another window of size $n$ is erased completely.
	
	If one would choose $|J|=2$, we would be able to correct $8$ erasures but all these erasures have to be in positions in
	$$
\bigcup_{t=1}^{L+1}\left(\{1+(t-1)n,\hdots,n+(t-1)n\}\setminus\{J+(t-1)n\}\right)
$$
whereas no erasures in
$$
\bigcup_{t=1}^{L+1}\left(\{1+(t-1)n,\hdots,n+(t-1)n\}\cap\{J+(t-1)n\}\right)
$$
could be corrected.

If one would choose $|J|=4$, we can again correct $8$ erasures in the case that correction with all $\mathcal{C}_t$ is possible but only $2$ of these erasures can have a position in
	$$
\bigcup_{t=1}^{L+1}\left(\{1+(t-1)n,\hdots,n+(t-1)n\}\setminus\{J+(t-1)n\}\right).
$$
Moreover, we could correct up to $12$ erasures in the case that the first window of size $n=10$ is lost completely and up to $14$ erasures in the case that another window of size $n$ is erased completely.

	Again our choice of $J$ is optimal if we assume the erasures to occur independently of each other.
\end{example}

\begin{remark}
The major advantage of using convolutional codes instead of block codes is that the symbols in different windows of size $n$ are dependent on each other and hence erasures cannot only be recovered with the help of the received symbols in the same window but also with the help of received symbols of other windows. This is illustrated also by the previous examples where recovery is possible if all symbols with positions in $J$ are erased if not too many symbols with positions in $\{J+n\}$ and $\{J+2n\}$ are erased.
This is due to the fact that there are erasure patterns where all symbols of the first window of size $n$ are erased but recovery with a convolutional code is still possible. But of course. this can never be possible using block codes since in this case all windows of size $n$ have to be decoded independently of each other.
\end{remark}

\section{Construction of suitable streaming codes}

The aim of this section is to provide constructions for $(n,k,\delta)$ MDP convolutional codes $\mathcal{C}$, which have the additional property that, for $f=1,\hdots,\mu$, $\mathcal{C}_f$ is an $[n,(f+1)k]$ MDS block code, as proposed at the beginning of the previous section. To this end, we will use the following lemma and proposition.

\begin{lemma}\label{lem:MDSblock}\cite{ma77}
	Let $\C$
	be an $[n,k]$ block code with generator matrix $G$. Then, $\C$ is MDS if, and only if, all $k \times k$ full size minors of $G$ are nonzero.
\end{lemma}

\begin{proposition}\cite[Theorem 3.3]{al16}\label{prop}	Let $\alpha$ be a primitive element of a finite field $\mathbb F=\mathbb F_{p^N}$ and $B=[b_{i,l}]$ be a matrix over $\mathbb F$ with the following properties
	\begin{enumerate}
		\item if $b_{i,l}\neq 0$, then $b_{i,l}=\alpha^{\beta_{i,l}}$ for a positive integer $\beta_{i,l}$
		\item if $b_{i,l}=0$, then $b_{i',l}=0$ for any $i'>i$ or $b_{i,l'}=0$ for any $l'<l$
		\item if $l<l'$, $b_{i,l}\neq 0$ and $b_{i,l'}\neq 0$, then $2\beta_{i,l}\leq\beta_{i,l'}$
		\item if $i<i'$, $b_{i,l}\neq 0$ and $b_{i',l}\neq 0$, then $2\beta_{i,l}\leq\beta_{i',l}$.
	\end{enumerate}
	Suppose $N$ is greater than any exponent of $\alpha$ appearing as a nontrivial term of any minor of $B$. Then $B$ has the property that each of its minors which is not trivially zero is nonzero.
\end{proposition}

The following theorem gives the desired construction.

\begin{theorem}
	Let $p$ be prime, $N\in\mathbb N$ and $\alpha$ be a primitive element of $\mathbb F_{p^N}$. For $i=1,\hdots,\mu$, set
	\begin{align}
	G_i:=\left[\begin{array}{ccc} \alpha^{2^{in}} & \cdots &  \alpha^{2^{(i+1)n-1}}\\
	\vdots & & \vdots\\
	\alpha^{2^{in+k-1}} & \cdots &  \alpha^{2^{(i+1)n+k-2}} \end{array}\right].
	\end{align}
	Then, the convolutional code $\mathcal{C}$ with generator matrix $G(z)=\sum_{i=0}^{\mu}G_iz^{i}$ is an MDP convolutional code and moreover, for $0\leq t\leq \mu$, $\begin{pmatrix}G_0\\ \vdots\\ G_t\end{pmatrix}$ is the generator matrix of an MDS block code if $N>\max\{2^{n(L+2)-1}, 2^{(\mu+1)n+k-1}\}$.
\end{theorem}

\begin{proof}
Obviously, the fullsize minors of $\left[\begin{array}{ccc} G_0 & & \\ \vdots & \ddots & \\ G_L & \cdots & G_0\end{array}\right]$ and $\left[\begin{array}{ccc} & & G_0\\ & \text{\reflectbox{$\ddots$}} & \vdots\\ G_0 & \cdots & G_L\end{array}\right]$  are equal. Thus, it follows from Theorem \ref{slid} and Proposition \ref{prop} that $\mathcal{C}$ is an MDP convolutional code if $N>2^{n(L+2)-1}$ (for the bound on $N$ also see Theorem 3.2 of \cite{al16}). Moreover, it follows from Lemma \ref{lem:MDSblock} and Proposition \ref{prop} that $\begin{pmatrix}G_0\\ \vdots\\ G_t\end{pmatrix}$ for $0\leq t\leq \mu$ are generator matrices of MDS block codes if $N>2^{(\mu+1)n+k-1}>\sum_{j=(\mu+1)n-1}^{(\mu+1)n+k-2}2^j$.
\end{proof}

\section{Conclusion}


We have studied the problem of private streaming of a sequence of files having the resilience against unresponsive servers the primary metric for judging the efficiency of a PIR scheme. We proposed for the first time a general scheme for such a problem. This scheme is based on MDP convolutional codes and the star product of codes.
It suits for a context where some servers fail to respond in contrast to other solutions considered in the literature where all the servers were assumed to fail at the same time instant. The approach presented can retrieve files in a sequential fashion and therefore is optimal for low-delay streaming applications. Some examples were presented to show how to take advantage of the proposed scheme.
We derived a large set of erasure patterns that our codes can recover. Concrete constructions of such codes exist although large field sizes are required. The construction of optimal codes for PIR over small fields that can deal with both burst and isolated erasures/errors is an interesting open problem that requires further research.

\section*{Acknowledgment}

The work of the first and third author was supported by the Portuguese Foundation for Science and Technology (FCT-Funda\c{c}\~{a}o para a Ci\^{e}ncia e a Tecnologia), through CIDMA - Center for Research and Development in Mathematics and Applications, within project UID/MAT/04106/2019. The first author was supported by the German Research Foundation within grant
LI 3101/1-1. The second author was partially supported by Spanish grant AICO/2017/128 of the Generalitat Valenciana and the University of Alicante under the project VIGROB-287.

\bibliographystyle{plain}
\bibliography{biblio_com_tudo}

\end{document}